\documentclass[letterpaper,11pt]{article}
\usepackage{amsmath,amssymb}
\usepackage{mathtools}
\usepackage{algorithm, algorithmic, float}
\newcommand{\x}{\otimes}
\newcommand{\Ceg}{\mathcal{C}}
\newcommand{\Ueg}{\mathcal{U}}
\newcommand{\Seg}{\mathcal{S}}

\newcommand{\Reg}{\mathsf{R}}

\newcommand{\Step}[1]{Step.~\hyperref[step:#1]{\ref*{step:#1}}}
\newcommand{\Heg}{\mathrm{HT}}

\newcommand{\spn}{\mathrm{span}}

\newcommand{\ket}[1]{|#1\rangle}
\newcommand{\bra}[1]{\langle#1|}
\newcommand{\braket}[2]{\langle#1|#2\rangle}

\DeclarePairedDelimiter{\ceil}{\lceil}{\rceil}

\DeclarePairedDelimiter{\abs}{\lvert}{\rvert}
\usepackage{array, multirow, booktabs}

\usepackage{authblk}
\usepackage[margin = 1in]{geometry}
\usepackage{graphicx}
\graphicspath{{Figures/}{logo/}}
\usepackage{amsmath,amssymb}
\usepackage{mathtools}
\usepackage{algorithm, algorithmic, float}
\usepackage{tikz}
\usepackage{amsthm}

\newtheorem{definition}{Definition}
\theoremstyle{plain}
\newtheorem{remark}[definition]{Remark}
\newtheorem{theorem}{Theorem}

\newtheorem{lemma}{Lemma}
\usepackage[colorlinks = true]{hyperref}
\usepackage{cleveref}

\title{Improvement of quantum walk-based search algorithms in single marked vertex graphs}
\author[a,b]{Xinying Li}
\author[a,c]{Yun Shang \thanks{shangyun@amss.ac.cn}}
\affil[a]{Institute of Mathematics, Academy of Mathematics and Systems Science, Chinese Academy of Sciences, Beijing 100190, China}
\affil[b]{School of Mathematical Sciences, University of Chinese Academy of Sciences, Beijing 100049, China}
\affil[c]{NCMIS, MDIS, Academy of Mathematics and Systems Science, Chinese Academy of Sciences, Beijing 100190, China}
 
\date{}
\begin{document}
	\maketitle
	
	\abstract{Quantum walks are powerful tools for building quantum search algorithms or quantum sampling algorithms named the construction of quantum stationary state.
		However, the success probability of those algorithms are all far away from 1. Amplitude amplification is usually used to amplify success probability, but the souffl\'e problems follow. Only stop at the right step can we achieve a maximum success probability. Otherwise, as the number of steps increases, the success probability may decrease,  which will cause troubles in practical application of the algorithm when the optimal number of steps is not known. 
		
		In this work, we define generalized interpolated quantum walks, which can both improve the success probability of search algorithms and avoid the souffl\'e problems.  
		
	Then we combine generalized interpolation quantum walks with quantum fast-forwarding. The combination both reduce the times of calling walk operator of searching algorithm from $\Theta((\varepsilon^{-1})\sqrt{\Heg})$ to $\Theta(\log(\varepsilon^{-1})\sqrt{\Heg})$ and reduces 
the number of ancilla qubits required from $\Theta(\log(\varepsilon^{-1})+\log\sqrt{\Heg})$ to $\Theta(\log\log(\varepsilon^{-1})+\log\sqrt{\Heg})$,
	and the souffle problem is avoided while the success probability is improved, where $\varepsilon$ denotes the precision and $\Heg$ denotes the classical hitting time.

	Besides, we show that our generalized interpolated quantum walks can be used to improve the construction of quantum states corresponding to stationary distributions as well.
		
		Finally, we give an application that can be used to construct a slowly evolving Markov chain sequence by applying generalized interpolated quantum walks, which is the necessary premise in adiabatic stationary state preparation.}


	\maketitle



	\section{Introduction}\label{sect:intro}
	Algorithms based on quantum walks have been studied widely and  speed-up classical random walk algorithms on many problems, such as element distinctness \cite{Ambainis2004QuantumWA}, triangle finding \cite{Magniez2005QuantumAF}, matrix product verification \cite{Buhrman2006QuantumVO}, especially on searching problems and sampling problems.
	
	In classical random walks, hitting time describes the speed to reach marked vertex set and mixing time denotes the convergence speed of Markov chain to approximate stationary distribution. They reflect the algorithm complexity of classical search problems and sampling
	
	problems and have many applications \cite{Sheng2019LowMeanHT}, such as mixing time can be used for Google's PageRank algorithm of ranking websites \cite{Page1999ThePC}.
	
	Since quantum walks have shown their advantages compared with random walks, the studies in quantum version of hitting time and mixing time is the current research focus.
	By quantizing classical discrete Markov chains, a bipartite quantum walk model was proposed by Szegedy \cite{sze}, which proved that quantum algorithms quadraticly speed up classical hitting time in detecting the existence of marked vertices for symmetric Markov chain. After Szegedy's framework, there appeared many new algorithms with better complexity performance in many variety contexts \cite{Szegedy2004SpectraOQ, Wocjan2008SpeedupVQ, Yung2012AQM}.
	
	By introducing recursive amplitude amplification, Magniez et al. extended Szegedy's framework to find marked vertex in reversible ergodic Markov chains \cite{Magniez2007SearchVQ}. However, its complexity can not  quadratic speed-up the classical algorithms such as for 2D grid search problem. Then Tulsi proposed a new technique to solve the open problem \cite{Tulsi2008FasterQA}. By extending Tulsi's technique, Magniez et al. \cite{MNRS12} demonstrated the possibility of finding a unique marked vertex with a square root speed up over the classical hitting time for any reversible state-transitive single-marked vertex Markov chain.  In \cite{Krovi2015QuantumWC},  
	Krovi et al. introduced the interpolated quantum walks, the serious restrictions are relaxed to reversible Markov chain with single-marked vertex. Recently, Ambainis et al. \cite{Ambainis2020QuadraticSF} introduced a search algorithm for any reversible Markov chain with quadratic speedup, which is more suitable for multi-marked vertices graphs, but does not perform as well as \cite{Krovi2015QuantumWC} on single-marked vertex graphs.

	Simultaneously, reverse search algorithms and stationary states can be prepared in $\Theta(\sqrt{\Heg}\frac{1}{\varepsilon})$ walk steps for any reversible Markov chain \cite{Krovi2015QuantumWC}. Further, the number of walk steps and ancilla qubits required is improved in \cite{our} by constructing a new reflection on stationary state and introducing quantum fast-forwarding.

	However, all above search algorithms and qsampling algorithms can only achieve constant success probability or even less, and in order to improve success probability, amplitude amplification is usually introduced. Unfortunately, the implementation of amplitude amplification often introduces the souffl\'e problems \cite{souffle}, because the largest success probability can be achieved only when we stop at the right time. Stopping too early or too late will not maximize the success probability. Since the success probability has only a lower bound instead of an exact value, which means we do not know the proper time to stop. When we stop late, the success probability will decrease as the number of steps increases. 
	Besides, now all search algorithms with constant success probability are linearly related to $\varepsilon^{-1}$, where $\varepsilon$ denotes the error of algorithm. Obviously, the dependency on error is too large. For example, if we want to improve the success probability from $1/16$ to $1/4$ with error $1/400$, the cost of current algorithm will become 100 times of the original algorithm cost \cite{Krovi2015QuantumWC}. When it comes to the qsampling algorithm, the error required is even smaller, and the number of calling walk operator will follow with a sharp increase.
	
	To increase success probability, in this work, we introduce generalized interpolated quantum walks, which will be combined with phase estimation to construct a new search algorithm instead of amplitude amplification. The introduction of generalized interpolated quantum walks both amplifies the success probability and avoids the souffl\'e problems while maintaining the time complexity. 
	
	Based on reversible Markov chains, quantum fast-forwarding algorithm \cite{qff} can simulate actions of $n$ random walk steps by $\sqrt{n}$ quantum walks steps.
	Here we find generalized interpolated quantum walks can be combined with quantum fast-forwarding. The combination reduces the dependency on error, which not only reduces the time complexity but also reduces the number of ancilla qubits required. 
	
	Similar to the above search algorithm, for qsampling problems, our generalized interpolated quantum walks also improve the success probability and avoid the souffl\'e problems as well. 
	
	Finally, we give an application of generalized interpolated quantum walks. 
	In adiabatic quantum computing, a quantum state can be prepared by preparing the ground state of a slowly evolving Hamiltonian sequence, and the ground state of Hamiltonian can be seen as the stationary state of corresponding Markov chain. However, it is not easy to prepare the needed Hamiltonian or Markov chain in adiabatic quantum computing.
	By applying generalized interpolated walks, we can construct a series of slowly evolving  Markov chains and prepare quantum stationary state in adiabatic quantum computation.
	
	The paper is organized as follows. 
	Preliminaries is provided in \Cref{sect:preliminaries} firstly. 
	Second we define generalized interpolated walks in \Cref{sect:generalized}. 
	Then to amplify the success probability of search algorithms and qsampling algorithms, we  apply generalized interpolated walks with quantum phase estimation and quantum fast-forwarding instead of amplitude amplification in \Cref{sect:QEE} and \Cref{sect:Qsampling}. Besides, we apply generalized interpolated walks to prepare quantum stationary state of Markov chains in adiabatic quantum computing in \Cref{sect:Application}. Finally, the paper is concluded in \Cref{sect:Discussion}. 
	
	\section{Methods}\label{sect:methods}
	\subsection{Preliminaries}\label{sect:preliminaries}		
	We introduce random walks firstly. Consider a graph $G(V, E)$ consisting of $N = |V|$ vertices, the transition probability between vertices depends on Markov chain $P$. For arbitrary vertex $x,y \in V$, $p_{xy}$ denotes the transition probability from $x$ to $y$. $P=(p_{xy})_{x,y \in V}$ is the corresponding transition matrix of the Markov chain, and the corresponding discriminant matrix is
	$$D(P) \coloneqq \sqrt{P \circ P^T},$$
	where '$\circ$' and square root are computed element-wise. 
	
	For any Markov chain P, if after enough steps every vertex in it can reach any other vertex, and the length of each directed cycle has the greatest common factor of 1, then P is ergodic. Ergodic Markov chain has a unique $1$-eigenvector, $i.e.$ stationary distribution, which is represented by $\pi ={(\pi_x)}_{x\in V}$. A chain is reversible if it satisfies $\pi_xp_{xy} = \pi_yp_{yx}$, and naturally, there is $D(P)\pi = \pi$.
	
	Assume P is a reversible Markov chain. 
	By introducing lazy walks \cite{Krovi2015QuantumWC}, all eigenvalues of $P$ lie in $[0,1]$,  where $\lambda_0$ denotes the eigenvalue of $P$ that equals to $1$ and the other  eigenvalues $\lambda_j$ for $j = 1, \dotsc, n-1$ are strictly less than $1$ in non-increasing order.
	The eigenvalue gap of $P$ is $\Delta \coloneqq 1 - \lambda_1$.

	For any ergodic reversible Markov chain P  on graph $G(V, E)$, the corresponding quantum walks perform on the extended space $\Lambda_V\x\Lambda_V$, where $\Lambda_V\coloneqq\spn\{\ket{x},x\in V\}$. For arbitrary initial probability distribution $\sigma$ over $V$, the initial state is defined as $\ket{\sigma}\ket{\bar{0}}$,  where ``$\ket{\bar{0}}$'' in the second register is some fixed initialization state in $\Lambda_V$.	The corresponding quantum walk operator of reversible Markov chain P  can be expressed as $$W(P) \coloneqq V(P)^{\dagger} \cdot S \cdot V(P) \cdot R_{\bar{0}},$$ where $V(P)$ is derived from the transition matrix $P$ as
	\begin{equation}
		\nonumber
		V(P) \ket{x} \ket{\bar{0}}\coloneqq \ket{x} \sum_{y \in X} \sqrt{P_{xy}} \ket{y}.
	\end{equation}
	The unitary swap operator $S$ and reflection operator $R_{\bar{0}}$ are expressed respectively as
	$$	\ket{x}\ket{y}\mapsto S\ket{x}\ket{y} \coloneqq
	\begin{cases}
		\ket{y}\ket{x}, & \text{if } (x,y) \in E, \\
		\ket{x}\ket{y}, & \text{otherwise.}
	\end{cases}$$
	$$	\ket{x}\ket{y}\mapsto R_{\bar{0}}\ket{x}\ket{y} \coloneqq I\x \left( 2\ket{\bar{0}}\bra{\bar{0}} - I\right)\ket{x}\ket{y}.$$
	
	For a graph $G$ with a marked subset $M \subseteq V$, the expectation time to hit any marked vertex in $M$ by a random walk with Markov chain P is called classical hitting time $\Heg(P, M)$, which is defined as 
	\begin{equation}\label{eq:ht}	\Heg(P, M) \coloneqq \sum_{k=1}^{n-|M|}\frac{|\langle v_k'| \bar{\pi} \rangle|^2 }{1-\lambda_k' },	\end{equation}
	where $\lambda_k'$ and $\ket{v_k'}$ are eigenvalues and the corresponding eigenvectors of the discriminant matrix $D(P')$ defined by the Markov chain $P'$ in nondecreasing order, and $P'$ is the absorbing  version of $P$ that all outgoing transformations from marked vertices are replaced by self-loops. Let $\Heg\coloneqq\max_{x\in V}\Heg(P,\{x\})$ denote the maximum hitting time to reach any vertex in $V$.
	
	In order to search a vertex in $M$, Krovi et al. introduced interpolation into Markov chain firstly \cite{Krovi2015QuantumWC, Krovi2010AdiabaticCA}. The interpolated Markov chain $P(s)$ can be expressed as
	{\tiny }		$$P(s) \coloneqq (1-s)P + sP',$$
	where $0\le s \le 1 $, and the corresponding $W(s)$ and $D(s)$ are defined as $W(P(s))$ and $D(P(s))$.
	
	From the construction of $W(s)$, there may be more than one $1$-eigenvectors for $W(s)$. Here we choose one of them as $\ket{v_0(s)}$ to prove the following theorem. Define $\ket{v_0(s)}$ as
	\begin{equation}\label{eq:v0}
		\ket{v_0(s)} =  \sum_{x}\sqrt{\frac{(1-s)\pi_x+s\pi'_x}{1-s+s\pi_M}}\ket{x},
	\end{equation}
	where $\pi'$ is a classical vector that all zeros in the unmarked vertices, and the same as $\pi$ in the marked vertices. It is easy to get 
	$$W(s)\ket{v_0(s)}\ket{\bar{0}}= \ket{v_0(s)}\ket{\bar{0}}, \ D(s)\ket{v_0(s)}= \ket{v_0(s)}, \quad  \forall s \in \left[ 0, 1\right].$$

	Similar to the classical random walk algorithm, assume that quantum walks have access to the following (controlled) unitary operators:
	
	$Setup(\pi)$: generates $\ket{\pi}\ket{\bar{0}}$, where$ \ket{\pi}=\sum_x\sqrt{x}$. Complexity $\Seg$.
	
	$Update(P)$: implements a (controlled) walk operator $W(P)$. Complexity $\Ueg$.
	
	$Check(M)$: checks whether $x$ is a marked vertex. Complexity $\Ceg$. Described by the mapping
	$$\forall x\in V, b\in{0,1}:\ket{x}\ket{b}\to \left\{\begin{matrix}
		\ket{x}\ket{b} \quad \quad \ & if \  x\notin M\\
		\ket{x}\ket{b\oplus 1} & if\  x \in M.
	\end{matrix}\right.$$
	We list the search algorithm and related details as below.
	The search algorithm proposed by Krovi et al. \cite{Krovi2015QuantumWC} consists of interpolated walks and quantum phase estimation and can be expressed as follows:
	
	\ \ \ \ 1. Use $Setup(\pi)$ to prepare $\ket{\pi}\ket{\bar{0}}$.
	
	\ \ \ \ 2. Apply $Check(M)$ to check whether the current vertex belongs to $M$.
	
	\ \ \ \ \ \ \ \ a. If it is marked, measure and output the current state.
	
	\ \ \ \ \ \ \ \ b. Otherwise, apply $U_{qee}(W(s), \ceil{\log\sqrt{\Heg}})$.\\
	$U_{qee}(W(s), \ceil{\log\sqrt{\Heg}})$ is quantum phase estimation operator and  expressed in the following lemma.
	\begin{lemma}
	[Phase estimation~\cite{Cleve1998QuantumAR}]\label{lem:phase estimation}
		Let $A$ be an arbitrary unitary operator on $n$ qubits and $\ket{\Psi_k}$ is any eigenvector of $A$ with eigenvalue $e^{i\phi_k}$, where $0 \le \phi_k \le \frac{\pi}{2}$. It does not need to know the exact value of $A$ or $\ket{\Psi_k}$ or $e^{i\phi_k}$. For any precision $t \in \mathbb{N}$, there exists a phase estimation algorithm $U_{qee}(A, t)$, which performs controlled-$A$,  controlled-$A^2$, ..., controlled-$A^{2^{t-1}}$. 
		For $k = 0, \cdots, n-1$,  the phase estimation algorithm $U_{qee}(A, t)$ acts on the eigenvector $\ket{\Psi_k}$ as
		\begin{equation*}
			\ket{\Psi_k} \ket{0^t} \mapsto U_{qee}(A, t)\ket{\Psi_k} \ket{0^t} = \ket{\Psi_k} \frac{1}{2^t} \sum_{l, m = 0}^{2^t - 1} e^{-\frac{2\pi i l m}{2^t}} e^{i\phi_k l} \ket{m} =: \ket{\Psi_k} \ket{\xi_k},
		\end{equation*}
		where $\ket{\Psi_k}\ket{\xi_k}$ satisfies that $\langle{0^t}|{\xi_k}\rangle
		= \frac{1}{2^t} \sum_{l=0}^{2^t-1} e^{i \phi_k l}$. The algorithm calls $\mathcal{O}(2^t)$ controlled-$A$ operator and $\mathcal{O}(t)$ ancilla qubits in total.
		\end{lemma}
	
	For any reversible Markov chain with single-marked vertex $g$, the output state in above search algorithm has constant overlap with the target state $\ket{g}$, but the maximum success probability is about $1/4-\varepsilon$. Amplitude amplification is usually applied to enlarge the success probability.
	
	Qsampling algorithm \cite{our} consists of quantum interpolated walks and quantum fast-forwarding algorithm and can be expressed as follows:
	
	\ \ \ \ 1. Use $Setup(g)$ to prepare $\ket{g}$ for any vertex $g \in V$.
	
	\ \ \ \ 2. Apply $U_{qfs}(W(s), \log\ceil{\sqrt{\Heg}})$.\\
	$U_{qfs}(W(s), \log\ceil{\sqrt{\Heg}})$ is a unitary operator composed of the quantum fast-forwarding operator and can be expressed as:
	$$ccX_{2,4}X_{4}(U_{qff}^{\dagger}(W(s), \log\ceil{\sqrt{\Heg}}) \x I) ccX_{234,5} ( U_{qff}(W(s), \log\ceil{\sqrt{\Heg}})\x I)X_{4}ccX_{2,4}.$$
	Let $\Lambda_0 \coloneqq\spn\{ \ket{0}, \ket{1}\}$, and the above operator is applied in five registers as $\Lambda_V\x\Lambda_V\x\Lambda_0^{\tau}\x\Lambda_0\x\Lambda_0^{r}$. The symbols above are defined as below: for any operator $A$, $A_{x}$ represents the $x$th register which is applied with $A$; $cA_{x,y}$ represents the controlled-$A$ gate with  control register $x$ and target register $y$, where $A$ is applied when the $x$th qubit is $\ket{1}$; similarly, $ccA_{x,y}$ represents that the controlled gate $A$ is applied when the $x$th qubit is $\ket{0}$. 
	
	The state after the unitary operator $U_{qfs}$ has constant overlap with the target state $\ket{\pi}$, but the maximum possible value is not enough and amplitude amplification is usually applied to enlarge success probability. 
	In the following we propose a concept of generalized quantum interpolated walks, which can improve the success probability of above search algorithm and qsampling algorithm to at least  $4/5$.

	\subsection{Generalized quantum interpolated walks} \label{sect:generalized}
	Here we define generalized interpolated walks as follows.
	
	\begin{definition}[Generalized quantum interpolated walks]
	Let $P$ be a Markov chain with a set of marked vertices $M$, and the corresponding absorbing walk is $P'$. For a series of parameter $\{s_1, s_2,. . .,s_r: 0\le s_1 , \cdots , s_r \le 1\} =: S$, an independent Markov chain $Q$ is defined to express the transition probability $q_{ij}$ from $s_i$ to $s_j$.  $s_{k_1},s_{k_2},\cdots$ is a sequence determined by the Markov chain $Q$, where $s_{k_{i+1}}$ is the next item of $s_{k_i}$ according to  $(q_{ij})_{1\le i,j \le r}$. Generalized interpolated Markov chain is defined as $P(S) \coloneqq \{P(s_{k_1}), P(s_{k_2}), \cdots \}$, where
		$$P(s_{k_i}) =(1-s_{k_i})P + s_{k_i}P'.$$
	\end{definition}
	
	All the previous quantum interpolated walks can be seen as special cases of generalized interpolated walks. We summarize and list them in \Cref{tab:1} and apply a special case in following sections.
	
	\begin{figure}[H]
		\centering 
		\includegraphics[width=0.7\textwidth,trim=10 0 20 0,clip]{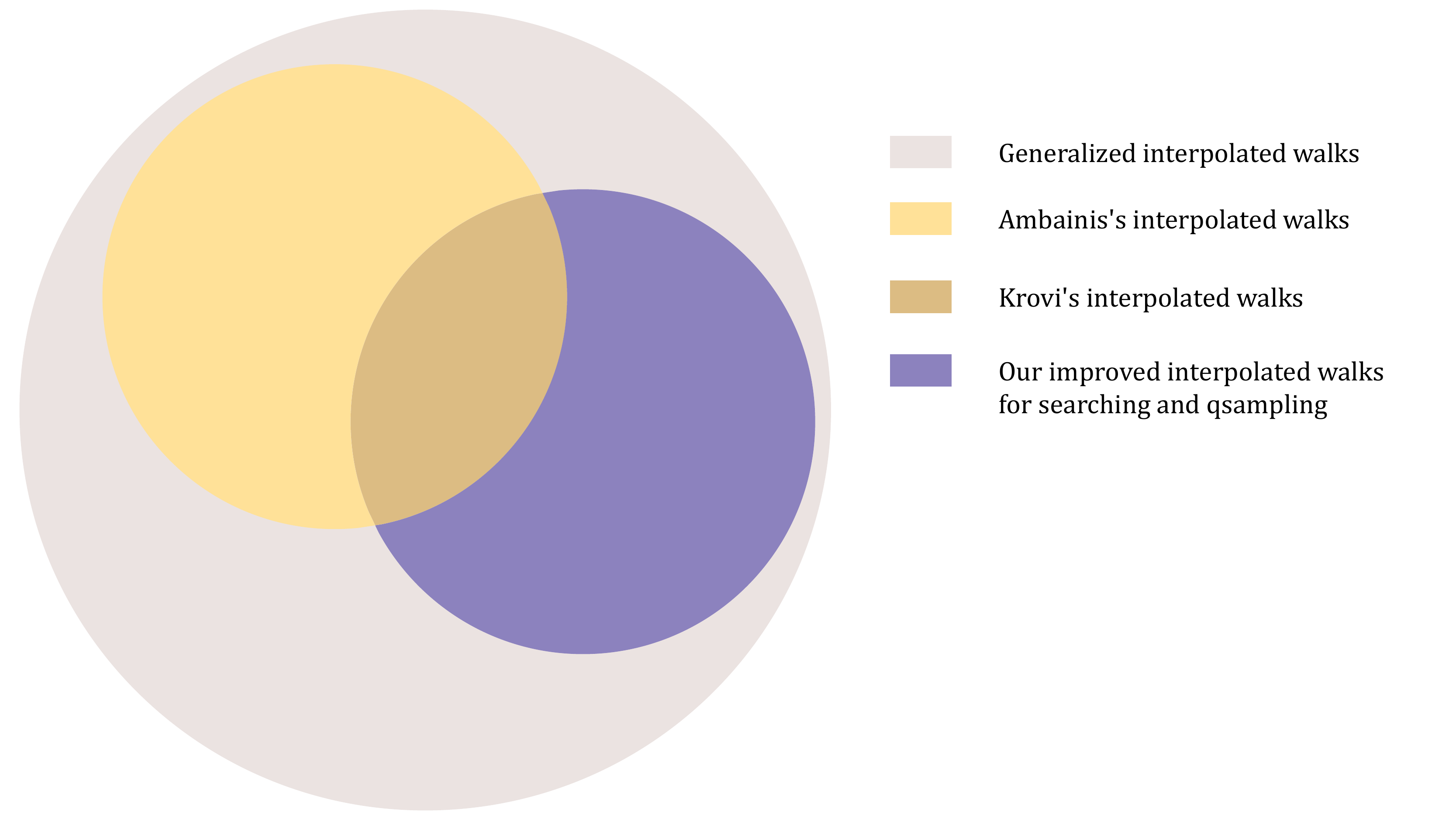}
		\caption{Relationship between all above interpolated walks.}
		\label{fig:0}
	\end{figure}
	\begin{table}[h!]
		\footnotesize
		\caption{The generalized interpolated walks and special cases}
		\label{tab:1}
		\tabcolsep 49pt
		\setlength{\tabcolsep}{7mm}{
			\begin{tabular*}{\textwidth}{cc}
				\toprule
				Quantum interpolated walks frameworks & $Q$ \\ \hline
				Generalized interpolated walks& $S=\{s_1, s_2,. . .,s_r\}$, $Q = (q_{ij})_{1\le i,j\le r}$  \\ 
				Krovi's interpolated walks \cite{Krovi2015QuantumWC}  & $r = 1$ \\
				Ambainis's interpolated walks \cite{Ambainis2020QuadraticSF} & $q_{i,j}= 1/r, \forall 1\le i,j\le r$ \\ 
				Our improved interpolated walks for searching and qsampling &$q_{i,i+1}= 1, \forall 1\le i\le r-1 $  \\ 
				\bottomrule
		\end{tabular*}}
	\end{table}	
	
	\begin{remark} [The relationship between interpolated walk framework]
		We list all above interpolated walks as follow, and the relationship between them is in \Cref{fig:0}.
		If $r = 1$, generalized quantum interpolated walks degenerate to Krovi's interpolated walks \cite{Krovi2015QuantumWC}, which are determined by only one parameter $s$ chosen by preparing a 1-eigenvector of $W(s)$ that has a large overlap with both the initial state and the target state.

		Ambainis's interpolated walks \cite{Ambainis2020QuadraticSF} can be seen as that each $s_i$ in $S\coloneqq\{s_1, s_2,\cdots,s_r: 0\le s_1, \cdots, s_r \le 1\}$ is chosen with equal probability and $W(s_i)$ acts on the initial state.
		Therefore the walk based on Markov chain Q is in complete graph, and the corresponding distribution keeps $(\frac{1}{r}, \frac{1}{r}, \cdots, \frac{1}{r})$. The series of $s_i$ can be constructed as $\{1-\frac{1}{r}: r \in \{1, 2, 4, \cdots, 2^{\ceil{\log(2592\Heg)}}\}\}$, and the corresponding walk operators are applied to the initial state $\ket{\pi}$ in superimposed form.
		
		Our improved interpolated walks which is applied in following section  fixe the order of each $s_i$ in $S\coloneqq\{s_1, s_2,\cdots,s_r: 0\le s_1, \cdots, s_r \le 1\}$. By implementing the walk operator corresponding to each $s_i$ one by one, as $i$ increases, the overlap between the 1-eigenvector of the current quantum walk operator and the target state will become larger.
		
		As far as we know, quantum interpolated walks that applied previously to search problems or qsampling problems have only one parameter $s$, and the maximum successful probability corresponding to single implementation is a small constant \cite{Krovi2015QuantumWC,Dohotaru2017ControlledQA} or even smaller \cite{Ambainis2020QuadraticSF}.
	\end{remark}
	
	\section{Results} \label{sect:algorithms}
	\subsection{Algorithms based on quantum phase estimation} \label{sect:QEE}
	By applying generalized interpolated walks instead of amplitude amplification, we will amplify the success probability of the search algorithms and qsampling algorithms  while maintaining the time complexity and avoiding the souffl\'e problems.
	\begin{figure}[H]
		\centering 
		\includegraphics[scale=0.3]{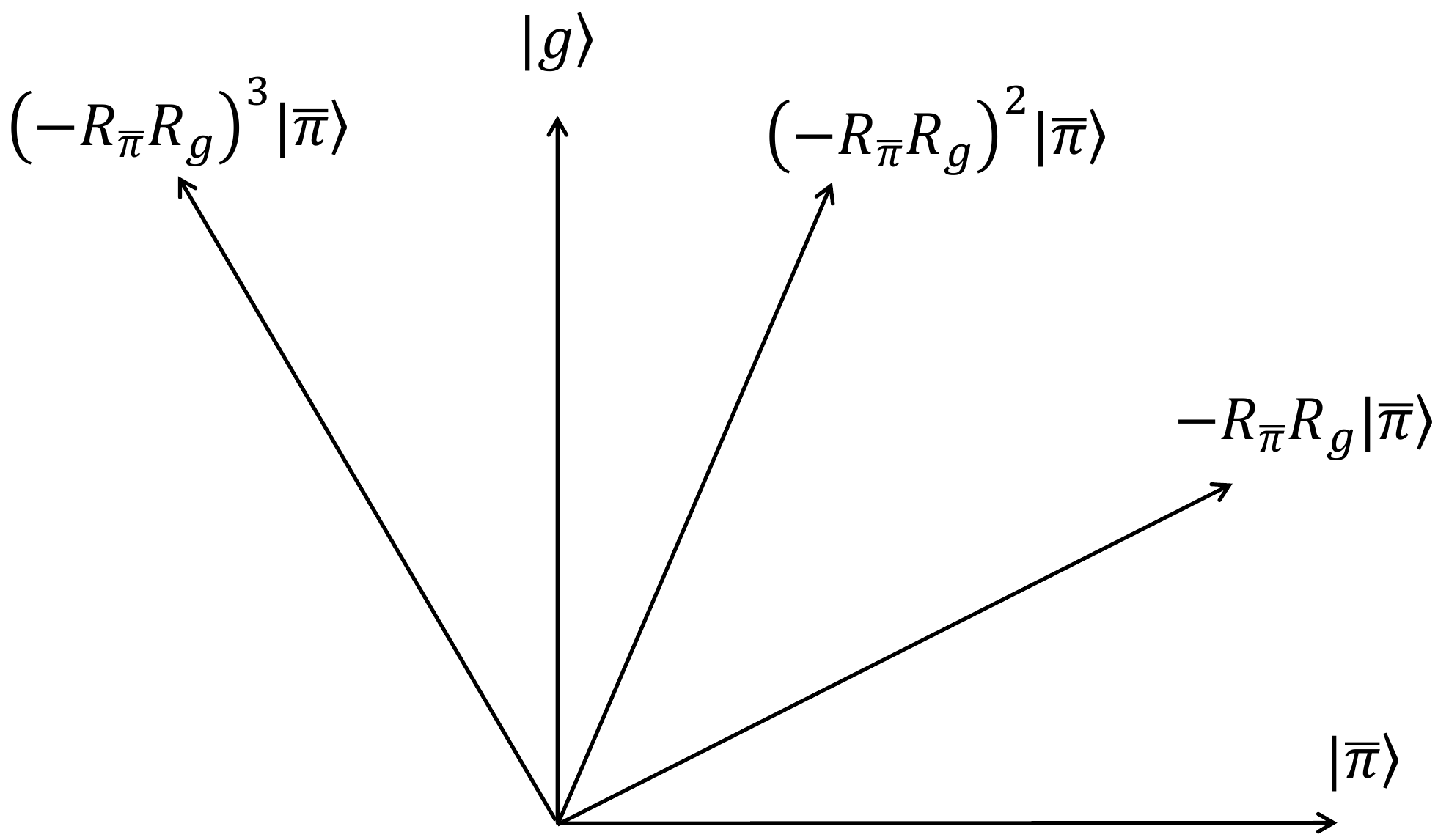}
		\caption{Souffl\'e problem.}
		\label{fig:1}
	\end{figure}
	Since the success probability of existing algorithms is still far from $1$, generally, amplitude amplification is used to enlarge the success probability, and then we often face souffl\'e problem.
	When the number of search steps is greater than the  steps  required, the success probability of these algorithms may drop dramatically as \Cref{fig:1}.
	Here we replace the single parameter $s$ with $S \coloneqq \{s_1, s_2,. . . ,s_n\}$, the set of parameters, and use this to amplify the success probability instead of amplitude amplification.
	\begin{theorem}
	[Search marked vertex based on quantum phase estimation]\label{thm:1}
		For an ergodic reversible Markov chain $P$ in $G(V,E)$ with single marked vertex $g$, \Cref{alg:1} achieves $\ket{g}$ with success probability more than $4/5$ from initial state $\ket{\pi}$. The complexity is $\Seg+\Theta(\varepsilon^{-1}\sqrt{\Heg}) (\Ueg +  \Ceg)$ with $\Theta(\log(\varepsilon^{-1})+\log\sqrt{\Heg})$ ancilla qubits, where $\varepsilon$ is error.
	\end{theorem}
	\begin{algorithm}[h!]
		\footnotesize
		\caption{The generalized interpolated search algorithm based on quantum phase estimation}
		\label{alg:1}
		\textbf{Input}: $\Gamma_1$, a number satisfies that $\Gamma_1 = \frac{r\pi\sqrt{\Heg}}{\sqrt{2}\varepsilon}$, where $\varepsilon$ is the error.\\
		\textbf{Output}: A state that has a constant overlap with the $\ket{g}$.
		\begin{algorithmic}[1]
			\STATE Set $\tau = \lceil  \log \Gamma_1 \rceil$, and prepare the initial state as $\ket{\pi}\ket{\bar{0}}\ket{0^{\tau}}$ on $\Reg_1\Reg_2\Reg_3$.
			\STATE Measure  $\Reg_1$ and check it whether it is $\ket{g}$ by $\Pi_g = \ket{g}\bra{g}$.\label{step:a}
			\IF {the current vertex is marked}
			\RETURN the marked vertex, stop.
			\ELSE
			\FOR {i = 1,2,...,r}
			\STATE Apply $U_{qee}(W(s_i), \tau)$ with $s_i$ to the current state.
			\ENDFOR
			\ENDIF
			\STATE Measure and output the first register.
		\end{algorithmic}
	\end{algorithm}
	
	Intuitively, our quantum algorithm works as \Cref{fig:2}. We fix the set of parameters as $\{s_1, s_2,\cdots,s_r: 0\le s_1\le \cdots \le s_r \le 1\}$. Firstly, we map $\ket{\bar{\pi}}$ to $\ket{v_0(s_1)}$ by quantum walk operator based on $P(s_1)$, and then map $\ket{v_0(s_1)}$ to $\ket{v_0(s_2)}$ by quantum walk operator based on $P(s_2)$, ..., finally, we map $\ket{v_0(s_{r-1})}$ to $\ket{v_0(s_r)}$ by quantum walk operator based on $P(s_r)$ and measure $\ket{v_0(s_r)}$ in the standard basis to get the marked vertex.	
	
	\begin{figure}[H]
		\centering 
		\includegraphics[scale=0.4]{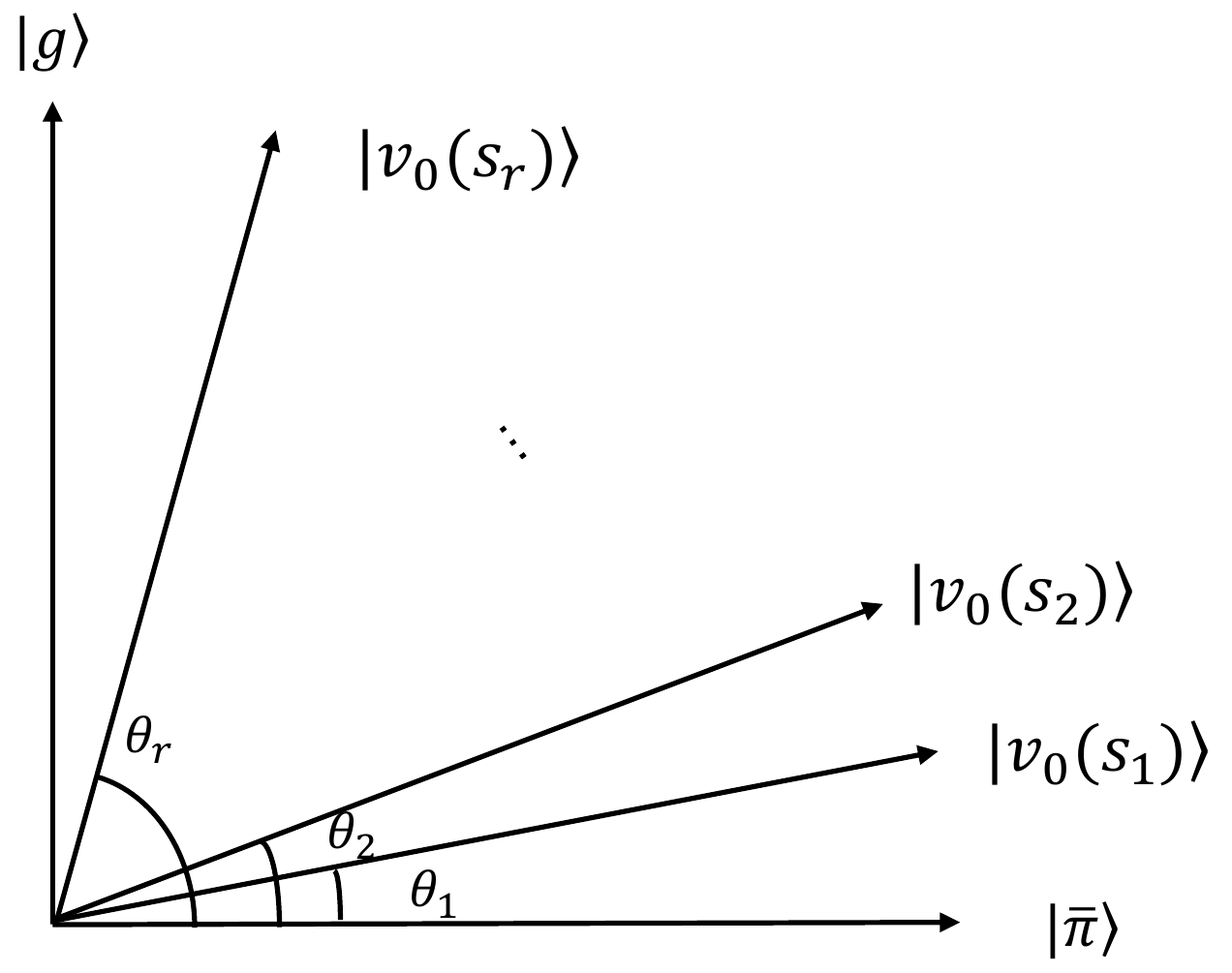}
		\caption{Our quantum algorithm framework.}
		\label{fig:2}
	\end{figure}
	We analyze the success probability firstly. 
	For the initial state $\ket{\pi}$ and target state $\ket{g}$, let $p_{succ}$ denote the success probability, then we have
	\begin{equation}
		\nonumber
		\sqrt{p_{succ}}		
		= \| \ket{g}\bra{g}_{1}U_{qee}(s_r)\cdots U_{qee}(s_2)U_{qee}(s_1)(\ket{\bar{\pi}}\ket{\bar{0}0^{\tau}}) \|.
	\end{equation}
	When $M=\{g\}$, define the state after \Step{a} as $\ket{\bar{\pi}}  \coloneqq \frac{1}{\sqrt{1-\pi_g}}\sum_{x \notin g}\sqrt{\pi_x}\ket{x}.$	From (\ref{eq:v0}), we have 	\begin{equation}	\nonumber
		\ket{v_0(s_i)}=\cos\theta_i\ket{\bar{\pi}}+\sin\theta_i\ket{g},	\end{equation}
	where $\cos\theta_i\coloneqq\sqrt{\frac{1-s_i}{s_i}}$, $\sin\theta_i\coloneqq\sqrt{\frac{2s_i-1}{s_i}}$.

	Set $\ket{v_0(s_0)} \coloneqq \ket{\bar{\pi}}$,	for $i  =1,\cdots,r$, $k  =0,\cdots,n-1$, and
	$\beta_k^i \coloneqq \left \langle  v_k(s_i) |  v_0(s_{i-1}) \right\rangle$. 
	Since for any $s \in [0,1] $ we have $\ket{\bar{\pi}} \in \Lambda_V = \spn\{\ket{v_k(s)}|k=0,\cdots,n-1\}$, then
	\begin{align*}
		\ket{\bar{\pi}}
		= \sum_{k=0}^{n-1} \left \langle  v_k(s_1) |  \bar{\pi} \right\rangle\cdot \ket{v_k(s_1)}
		=\sum_{k=0}^{n-1} \beta_k^1 \ket{v_k(s_1)}.
	\end{align*}
	Here  we give an estimation of the success probability $p_{succ}$ as below.	
	\begin{align}
		\sqrt{p_{succ}}	&=
		\|\ket{g}\bra{g}_{1} U_{qee}(s_r)\cdots U_{qee}(s_1)(\sum_{k=0}^{n-1} \beta_k^1 \ket{v_k(s_1)}) \ket{\bar{0}0^{\tau}}\|\tag*{} \\
		&\ge \|\ket{g} \bra{g}_{1}  \Pi_{0^{\tau}} U_{qee} (s_r) \cdots  \Pi_{0^{\tau}} U_{qee} (s_1)(\sum_{k=0}^{n-1} \beta_k^1 \ket{v_k(s_1)} ) \ket{\bar{0} 0^{\tau}}\|\tag*{where $\Pi_{0^{\tau}}=I\x I \x \ket{0^{\tau}}\bra{0^{\tau}}$ } \\
		&\ge \|\ket{g}\bra{g}_{1} \Pi_{0^{\tau}} U_{qee}(s_r)\cdots \Pi_{0^{\tau}}U_{qee}(s_1)\beta_0^1 \ket{v_0(s_1)}  \ket{\bar{0}0^{\tau}}\|-  \tag*{ } \\
		& \|\ket{g}\bra{g}_{1} \Pi_{0^{\tau}} U_{qee}(s_r) \cdots \Pi_{0^{\tau}} U_{qee}(s_1)(\sum_{k=1}^{n-1} \beta_k^1 \ket{v_k(s_1)}) \ket{\bar{0}0^{\tau}}\| \tag*{ from triangle inequality} \\
		&= \|\ket{g}\bra{g}_{1} \Pi_{0^{\tau}}U_{qee}(s_r)\cdots \beta_0^1 \ket{v_0(s_1)}  \ket{\bar{0}0^{\tau}}\| - A_1 \tag*{} \\
		&\ge \|\ket{g}\bra{g}_{1} \Pi_{0^{\tau}} U_{qee}(s_r) \beta_0^{r-1} \cdots \beta_0^1 \ket{v_0(s_r)}   \ket{\bar{0}0^{\tau}}\| -  A_1 - \cdots -A_r \tag*{} \\
		&= \|\beta_0^{r}\beta_0^{r-1}\cdots\beta_0^1\left \langle g |  v_0(s_r) \right\rangle \| - A_1 -\cdots -A_r,
		\label{eq:total1}
	\end{align}
	where $A_i \coloneqq \|\ket{g}\bra{g}_{1}  \Pi_{0^{\tau}} U_{qee} (s_r) \cdots  \Pi_{0^{\tau}} U_{qee} (s_{i}) ( \sum_{k=1}^{n-1} \beta_k^i\beta_0^{i-1} \cdots \beta_0^1  \ket{v_k(s_i)}) \ket{\bar{0}0^{\tau}} \| $.

	To amplify the success probability, we consider the first item of (\ref{eq:total1}) as follows:
	\begin{align*}
		&|\left \langle g |  v_0(s_r) \right\rangle\beta_0^{r}\beta_0^{r-1}\cdots\beta_0^1|\\
		=&|\cos(\pi/2-\theta_r)\cos(\theta_r-\theta_{r-1})\cdots \cos(\theta_2-\theta_1)\cos\theta_1|\\
		=:&G(r).
	\end{align*}
	Let $\theta_0=0$. 
	For $i=0,\cdots,r-1$, define $\vartheta_i \coloneqq \theta_{i+1}-\theta_{i}$  and then $$G(r)=\cos(\pi/2-\vartheta_{r-1}-\cdots -\vartheta_{0})\cos\vartheta_{r-1}\cdots \cos\vartheta_1\cos\vartheta_0.$$
	The problem to solve $G(r)$ becomes an optimization problem:
	$$\left\{\begin{matrix}
		X=\{\vartheta_0, \vartheta_1, \cdots, \vartheta_r\} \\
		s.t.\ \sum_i \vartheta_i=\pi/2, \ \vartheta_i\ge0 \quad for\ i \in \{0,1,\cdots,r\}\\
		\max G(X, r)
	\end{matrix}\right.$$
	For any $i \in \{0, 1, \cdots, r-1\}$, there is 
	\begin{align*}
		\frac{\partial G(X, r)}{\partial \vartheta_i}
		=&\frac{\partial [\cos(\pi/2-\vartheta_{r-1}-\cdots -\vartheta_{0})\cos\vartheta_{r-1}\cdots \cos\vartheta_1\cos\vartheta_0]}{\partial \vartheta_i}\\
		=&[ \sin( \pi/2 - \vartheta_{r-1} - \cdots - \vartheta_{0} ) \cos \vartheta_i - \cos( \pi/2 - \vartheta_{r-1} - \cdots - \vartheta_{0} ) \sin \vartheta_{i} ]\cdot \\
		&\cos\vartheta_{r-1}\cdots \cos\vartheta_{i+1}\cos\vartheta_{i-1}\cdots\cos\vartheta_0\\
		=& \sin(\pi/2-\vartheta_{r-1}-\cdots-\vartheta_{i+1}-2\vartheta_{i} -\vartheta_{i-1}-\cdots-\vartheta_{0})\cos\vartheta_{r-1}\cdot\\
		&\cdots\cos\vartheta_{i+1}\cos\vartheta_{i-1}\cdots\cos\vartheta_0.
	\end{align*}
	Define $\vartheta_{r}\coloneqq\pi/2-\sum_{j=0}^{r-1}\vartheta_{j}$. If $\vartheta_{i}<\vartheta_{r}$, $G(X, r)$ increases as $\vartheta_{i}$ increases; if $\vartheta_{i}>\vartheta_{r}$, $G(X, r)$ decreases as $\vartheta_{i}$ increases, which means the maximum of  $G(X, r)$ is achieved when $\vartheta_{i}=\vartheta_{r}$.
	Since $i$ is any number in $\{0, 1, \cdots, r-1\}$, we have $$ G(X, r)=G(\vartheta_0, \vartheta_1, \cdots, \vartheta_{r-1})
	\le G(\vartheta_r, \vartheta_r, \cdots, \vartheta_r).$$
	$$\pi/2=\sum_{j=0}^{r}\vartheta_{j}=(r+1)\vartheta_r,$$  that is $\vartheta_i=\frac{\pi}{2(r+1)}$ for $i  \in \{0, 1, \cdots, r-1\}$.
	Then we have	\begin{equation}	\max_{X}G(X, r)=\cos^{r+1}(\frac{\pi}{2(r+1)}),	\label{eq:theta}
	\end{equation} which is achieved when
	$$	\theta_i= \frac{i\pi}{2(r+1)}, \quad for \quad i \in \{0, 1, \cdots, r\}.	$$
	Now we bound the value of $A_i$. Define  $\alpha_{k,0}^i=\left \langle  v_k(s_{i}) |  \bar{\pi} \right\rangle$, $\alpha_{k,r+1}^i=\left \langle  v_k(s_i) | g\right\rangle$,
	since	\begin{align*}
		0 = \left \langle {v_k(s_i)} | {v_0(s_i)} \right \rangle
		= \cos\theta_i\left \langle {v_k(s_i)} |\bar{\pi} \right \rangle + \sin\theta_i\left \langle {v_k(s_i)} |g \right \rangle 
		=\cos\theta_i \alpha_{k,0}^i + \sin\theta_i \alpha_{k,r+1}^i,
	\end{align*}
	for $i  =1,\cdots,r$, $k  =1,\cdots,n-1$, we have $
	\alpha_{k,r+1}^i = -\frac{\cos\theta_i}{\sin\theta_i}\alpha_{k,0}^i$, then
	\begin{align}
		\beta_k^i=&	\left \langle  v_k(s_{i}) |  v_0(s_{i-1}) \right\rangle \tag*{}\\ 
		=&\cos\theta_{i-1}\left \langle  v_k(s_{i}) | \bar{\pi} \right\rangle +\sin\theta_{i-1}\left \langle  v_k(s_{i}) |  g \right\rangle \tag*{}\\
		=&\cos\theta_{i-1}\alpha_{k,0}^{i} +\sin\theta_{i-1}\alpha_{k,r+1}^{i}  \tag*{}\\
		=&\cos\theta_{i-1}\alpha_{k,0}^{i} -\sin\theta_{i-1}\frac{\cos\theta_i}{\sin\theta_i}\alpha_{k,0}^i \tag*{}\\
		=&[\cos\theta_{i-1} -\sin\theta_{i-1}\frac{\cos\theta_i}{\sin\theta_{i}}]\alpha_{k,0}^i \tag*{}\\
		=&[\cos(\frac{(i-1)\pi}{2(r+1)}) -\sin(\frac{(i-1)\pi}{2(r+1)})\frac{\cos(\frac{i\pi}{2(r+1)})}{\sin(\frac{i\pi}{2(r+1)})}]\alpha_{k,0}^i,
		\label{eq:9}
	\end{align}
	which means \begin{equation}
		\abs{	\beta_k^i}\le\abs{[\cos(\frac{(i-1)\pi}{2(r+1)}) -\sin(\frac{(i-1)\pi}{2(r+1)})\frac{\cos(\frac{i\pi}{2(r+1)})}{\sin(\frac{i\pi}{2(r+1)})}]}\cdot\abs{\alpha_{k,0}^i}\le\abs{\alpha_{k,0}^i}.
		\label{eq:ab}
	\end{equation}
	
	To determine the value of $A_i$, we introduce  the relationship between $W(s)$ and $D(s)$ as follows:
	\begin{lemma}
[Spectrum of $W(s)$ and $D(s)$~\cite{sze}]\label{lem:spectrum}
	The eigenvalues and eigenvectors $\ket{v_k(s)}$ of $D(s)$ satisfy that
		$$(D(s)\x I)\ket{v_k(s)}\ket{\bar{0}} = cos(\varphi_k(s))\ket{v_k(s)}\ket{\bar{0}},$$ for $k = 0, \dotsc, n-1$, with $\varphi_0(s) = 0$, the eigenvalues and eigenvectors of $W(s)$ are:
		$$W(s)\ket{\Psi_0(s)}  = \ket{\Psi_0(s)}, W(s)\ket{\Psi^\pm_k(s)} = e^{\pm i \varphi_k(s)}\ket{\Psi^\pm_k(s)},$$
		and the relation between them can be expressed as:
		$$\ket{\Psi_0(s)} := \ket{v_0(s)}\ket{\bar{0}}, \ket{\Psi^\pm_k(s)} \coloneqq\frac{\ket{v_k(s)}\ket{\bar{0}} \pm i \ket{v_k(s)}\ket{\bar{0}^\perp}}{\sqrt{2}}.$$
		Then $\mathcal{B}_k(s) = span\lbrace \ket{\Psi^+_k(s)}, \ket{\Psi^-_k(s)}\rbrace =  \spn\lbrace\ket{v_k(s)}\ket{\bar{0}}, \ket{v_k(s)}\ket{\bar{0}^\perp}\rbrace$ for $k = 1, \dotsc, n-1$ and $\mathcal{B}_0(s) = span\lbrace\ket{\Psi_0(s)}\rbrace= span\lbrace\ket{v_0(s)}\rbrace$ are invariant subspace of $W(s)$ and mutually orthogonal. Actually, the walk space of $W(s)$ is $\bigcup_{k=0}^{n-1} \mathcal{B}_k(s)$.
	\end{lemma}
	
	From \Cref{lem:phase estimation} and \Cref{lem:spectrum}, since $W(s_i)$ is a real operator, $U_{qee}(W(s_i),\tau)$ acts on the eigenvectors of $W(s_i)$ as:
	$$
	\ket{\Psi_0(s_i)}  \mapsto \ket{\Psi_0(s_i)} \ket{0^{\tau}},\quad
	\ket{\Psi^{\pm}_k(s_i)}  \mapsto \ket{\Psi^{\pm}_k(s_i)} \ket{\xi^{\pm}_k(s_i)},$$
	where $\ket{\xi^{\pm}_k(s_i)}$ is a ${\tau}$-qubit state that satisfies
	$		\left \langle 0^{\tau}  | \xi^{\pm}_k(s_i)  \right \rangle
	= \frac{1}{2^{\tau}} \sum_{l=0}^{2^{\tau}-1} e^{\pm i \varphi_k(s_i) l}
	=: \delta^{\pm}_k(s_i).$
	Thus, eigenvectors of $D(s_i)$ after phase estimation are
	\begin{align*}
		\nonumber
		\ket{v_k(s_i)}\ket{\bar{0}0^{\tau}} &\mapsto U_{qee}(s_i)\ket{v_k(s_i)}\ket{\bar{0}0^{\tau}}\\
		&=\frac{1}{\sqrt{2}} \sum_{k=1}^{n-1} 
		\bigl( \ket{\Psi^+_k(s_i)} \ket{\xi^+_k(s_i)} + \ket{\Psi^-_k(s_i)} \ket{\xi^-_k(s_i)} \bigr).
	\end{align*}
	Now we bound $A_i$ as
	\begin{align} A_i
		=&	 \|\ket{g} \bra{g}_{1} \Pi_{0^{\tau}} U_{qee}(s_r) \cdots  \Pi_{0^{\tau}} U_{qee}(s_{i})( \sum_{k=1}^{n-1} \beta_k^i\beta_0^{i-1} \cdots \beta_0^1\ket{v_k(s_i)} ) \ket{\bar{0}0^{\tau}}\|\tag*{}\\
		\le&	 \beta_0^{i-1} \cdots \beta_0^1\|\Pi_{0^{\tau}}U_{qee}(s_{i})(\sum_{k=1}^{n-1} \beta_k^i \ket{v_k(s_i)}) \ket{\bar{0}0^{\tau}}\|\tag*{}\\
		=&	  \beta_0^{i-1} \cdots \beta_0^1\|\Pi_{0^{\tau}} \frac{1}{\sqrt{2}} \sum_{k=1}^{n-1} \beta_k^i
		\bigl( \ket{\Psi^+_k(s_i)} \ket{\xi^+_k(s_i)} + \ket{\Psi^-_k(s_i)} \ket{\xi^-_k(s_i)} \bigr)\|\tag*{}\\
		=&	 \beta_0^{i-1} \cdots \beta_0^1 \| \frac{1}{\sqrt{2}} \sum_{k=1}^{n-1} \beta_k^i
		\bigl( \delta^+_k(s_i)\ket{\Psi^+_k(s_i)} + \delta^-_k(s_i)\ket{\Psi^-_k(s_i)}  \bigr)\|\tag*{}\\
		\le&	 \beta_0^{i-1} \cdots \beta_0^1\sqrt{\sum_{k=1}^{n-1} \abs{\beta_k^i}^2 \delta_k^2(s_i)} \le	 \sqrt{\sum_{k=1}^{n-1} \abs{\beta_k^i}^2 \delta_k^2(s_i)}.\tag*{$\ket{\Psi^\pm_1}, \dotsc, \ket{\Psi^\pm_k}$ are mutually orthogonal}
	\end{align}	
	
	where $\delta_k(s_i) := \abs{\delta^+_k(s_i)} = \abs{\delta^-_k(s_i)}$, ${\tau} = \ceil{\log \Gamma_1}$, and satisfies\footnote{The detailed construction can be seen in \cite{Krovi2015QuantumWC}.} $
	\delta_k^2(s_i) 	\leq \frac{\pi^2}{2^{2{\tau}} \varphi^2_k(s_i)} 	\leq \frac{\pi^2}{\Gamma_1^2 \varphi^2_k(s_i)}
	$.
	
	Similar to  (\ref{eq:ht}),  the interpolated hitting time in single-marked vertex graph is  
	\begin{equation}
		\Heg(s_i) = \sum_{k=1}^{n-1} \frac{\abs{\braket{v_k(s_i)}{\bar{\pi}}}^2}{1 - \lambda_k(s_i)}=   \sum_{k=1}^{n-1} \frac{\abs{\alpha_{k,0}^i}^2}{1 - \cos \varphi_k(s_i)}
		\geq 2 \sum_{k=1}^{n-1} \frac{\abs{\alpha_{k,0}^i}^2}{\varphi^2_k(s_i)}.
		\label{eq:16}
	\end{equation}
	
	By combining (\ref{eq:ab}) and (\ref{eq:16}) we have
	$$	\sqrt{\sum_{k=1}^{n-1} \abs{\beta_k^i}^2 \delta_k^2(s_i)}
	\le\sqrt{\sum_{k=1}^{n-1} \abs{\alpha_{k,0}^i}^2 \delta_k^2(s_i)}
	\le\frac{\pi}{\Gamma_1}\sqrt{\sum_{k=1}^{n-1} \frac{\abs{\alpha_{k,0}^i}^2}{\varphi^2_k(s_i)}}
	\le\frac{\pi}{\Gamma_1} \sqrt{\frac{\Heg(s_i)}{2}}.$$
	
	From theorem 17 in \cite{Krovi2015QuantumWC}, we have $HT(s_i) \le HT$ in the single-marked vertex Markov chain.
	Thus, $\frac{\pi}{\Gamma_1}\sqrt{\frac{\Heg(s_i)}{2}} \le \frac{\varepsilon}{r}$ and (\ref{eq:total1}) becomes
	\begin{equation}
		\nonumber
		\sqrt{p_{succ}}
		\geq \cos^{r+1}(\frac{\pi}{2(r+1)})  - r\cdot\frac{\varepsilon}{r}
		\geq \cos^{r+1}(\frac{\pi}{2(r+1)})  - \varepsilon,
	\end{equation}
	that is $p_{succ}\ge( \cos^{r+1}(\frac{\pi}{2(r+1)})  - \varepsilon)^2$.
	
	We evaluate how the value of $\cos^{r+1}(\frac{\pi}{2(r+1)})$ varies with $r$ as follows. 
	Let $f(r)=\cos^r(\frac{\pi}{2r})$ and $\ln(f(r))=r\ln(\cos(\frac{\pi}{2r}))$, then 
	\begin{align*}
		\frac{f'(r)}{f(r)}
		=&\ln(\cos(\frac{\pi}{2r}))+\frac{r[-\sin(\frac{\pi}{2r})\frac{\pi}{2}(-\frac{1}{r^2})]}{\cos(\frac{\pi}{2r})}\\
		=&\ln(\cos(\frac{\pi}{2r}))+\frac{\frac{\pi}{2r}\sin(\frac{\pi}{2r})}{\cos(\frac{\pi}{2r})}.
	\end{align*}
	
	Since $r \ge 2$, then $0<\frac{\pi}{2r}<\frac{\pi}{2}$ and $\cos(\frac{\pi}{2r})>0$,
	which means	to determine whether $f'(r)$ is positive or negative, we only need to determine  $$\cos(\frac{\pi}{2r})\ln(\cos(\frac{\pi}{2r}))+\frac{\pi}{2r}\sin(\frac{\pi}{2r})=:g(r).$$
	
	Let $y\coloneqq\frac{\pi}{2r}$ and $h(y)\coloneqq \cos(y)\ln(\cos(y))+y\sin(y)=g(r)$, then
	$$h'(y)\coloneqq -\sin(y)\ln(\cos(y))+\cos(y)\frac{-\sin(y)}{\cos(y)}+\sin(y)+y\cos(y).$$
	Since $r\ge2$, $y\in(0,1]$, $h'(y)\ge0$ and $h(y)$ is an increasing function, which means for any $r\ge2$, $g(r)=h(y)\ge h(0)=0$, that is $f'(r)\ge0$, $f(r)$ is an increasing function. 
	
	From the definition of $f(r)$, the success probability increases and tends to $1$ as the number of steps $r$ increases, which means  our method will never face the souffl\'e problems. When $r\ge 10$ and $\varepsilon$ is small, we have $p_{succ}\ge 4/5$, and the success probability increases as $r$ increases.

	\subsection{Algorithms based on quantum fast-forwarding} \label{sect:Qsampling}
	In the previous section, a new search algorithm based on generalized interpolated walks and phase estimation improves the success probability to nearly 1.  However, it relies on $\varepsilon^{-1}$ too heavily that when the value of $\varepsilon$ is small, the algorithm complexity also increases rapidly. Quantum fast-forwarding as a quantum version of random walk with acceleration can be used here to reduce the dependence on $\varepsilon$. In this section, we combine the generalized interpolated walks with quantum fast-forwarding, which both reduces the times of calling walk operators and the number of ancilla qubits.
	
	We introduce quantum fast-forwarding algorithm firstly.
	\begin{lemma}
[Quantum fast-forwarding~\cite{qff}]\label{lem:quantum fast-forwarding}
	For any reversible Markov chain and corresponding walk operator $W(P)$ on state space $\Lambda_V\x \Lambda_V$ and any $\ket{\psi} \in \Lambda_V$, quantum fast-forwarding algorithm $U_{qff}(W, \tau)$ acting on $\ket{\psi}\ket{\bar{0}0^{\tau}}$ will output a state $\varepsilon$-close to $D^{t}\ket{\psi}\ket{00^{\tau}} + \ket{\chi}\ket{{00}^{\perp}}$,
		where $\ket{{00}^{\perp}}$ satisfies $\langle 00^{\tau}|00^{\perp }\rangle =0$. The algorithm invokes controlled-$W$ operator $\Gamma \coloneqq \Theta( \sqrt{t\log(\varepsilon^{-1})}) $ times and requires  $\tau = \lceil \log\Gamma\rceil$ ancilla qubits.
	\end{lemma}
	
	The dynamics of discriminant matrix $D$ are simulated by the unitary operator $U_{qff}(W,\tau)$, which acts as follows:
	\begin{equation} \label{eq:close}
		\left \| \Pi_{\bar{0}0}U_{qff}(W,\tau)\ket{\psi}\ket{\bar{0}0^{\tau}} - (D^{t}\ket{\psi})\ket{\bar{0}0^{\tau}}  \right \| \le \varepsilon_{1},
	\end{equation}
	where $\Pi_{\bar{0}0} = I \x (\ket{\bar{0}}\bra{\bar{0}})\x(\ket{0^{\tau}}\bra{0^{\tau}})$.	The detailed construction of $U_{qff}(W,\tau) = V_{q}^{\dagger} W_{ctrl} V_{q}$ can be described as: $V_q \ket{\psi,\bar{0}}\ket{0^{\tau}}= \sum_{l=0}^{2^{\tau}-1} \sqrt{\frac{p_l}{1-(\sum_{x = {2^{\tau}}}^{t}p_x)}} \ket{\psi,\bar{0}}\ket{l}, W_{ctrl} = \sum_{l=0}^{2^{\tau}-1} W^l \x \ket{l}\bra{l}$, where $p_l$ is the probability of uniform random walk with a distance $l$ from the origin point at time $t$, the detailed description of $p_l$ can be found in Equation (6) of \cite{qff}. 
	
	Then the new algorithm for searching is described as follows.
	\begin{theorem}
[Search marked vertex based on quantum fast-forwarding]\label{thm:2}For an ergodic reversible Markov chain $P$ in $G(V,E)$ with single marked vertex $g$, \Cref{alg:qis} achieves $\ket{g}$ with success probability more than $4/5$ from initial state $\ket{\pi}$. The complexity is $\Seg + \Theta(\log(\varepsilon^{-1})\sqrt{\Heg}) (\Ueg +  \Ceg)$ with $\Theta(\log\log(\varepsilon^{-1})+\log\sqrt{\Heg})$ ancilla qubits, where $\varepsilon$ is the error.	
	\end{theorem}
	\begin{proof}
	\begin{algorithm}[H]
		\footnotesize
		\caption{The generalized interpolated search algorithms based on quantum fast-forwarding}
		\label{alg:qis}
		\textbf{Input}: $\Gamma_2$, a number satisfies that $\Gamma_2 = \Theta(\log(\varepsilon^{-1})\sqrt{\Heg})$, where $\varepsilon$ is the error; $S = \{s_1, s_2,. . . ,s_r\}$.\\
		\textbf{Output}: $\ket{\psi_r}$, the state that has a constant overlap with the $\ket{g}$.
		\begin{algorithmic}[1]
			\STATE Set $\tau = \lceil  \log \Gamma_2 \rceil$, prepare the initial state as $\ket{\pi}\ket{\bar{0}}\ket{0^{\tau}}\ket{0}\ket{0^r}$ on $\Reg_1\Reg_2\Reg_3\Reg_4\Reg_5$, where the five registers are the same as registers in \cite{our}.
			\STATE Measure  $\Reg_1$ and check whether it is $\ket{g}$.
			\IF {the current vertex is marked}
			\RETURN the marked vertex, stop.
			\ELSE
			\FOR {i = 1,2,...,r}
			\STATE Apply $U_{qfs}(W(s_i), \tau)$ with $s_i$  to the current state.
			\ENDFOR
			\ENDIF
			\STATE Measure and output the first register.
		\end{algorithmic}
	\end{algorithm}
	We analyze the success probability firstly. For initial state $\ket{\bar{\pi}}$ and target state $\ket{g}$, let $p_{succ}$ denote the success probability, then we have
	\begin{align}
		\sqrt{p_{succ}}		
		&= \| \ket{g}\bra{g}_{1}U_{qfs}(s_r)\cdots U_{qfs}(s_1)(\ket{\bar{\pi}}\ket{\bar{0} 0^{\tau} 0 0^r}) \|\tag*{} \\
		&= \|\ket{g}\bra{g}_{1}U_{qfs}(s_r)\cdots U_{qfs}(s_1)(\sum_{k=0}^{n-1} \beta_k^1 \ket{v_k(s_1)}) \ket{\bar{0}0^{\tau}00^r}\|\tag*{} \\
		&\ge \|\ket{g}\bra{g}_{1} \ket{1}\bra{1}_{\gamma_1 \cdots \gamma_r}U_{qfs}(s_r) \cdots  U_{qfs}(s_1)\beta_0^1 \ket{v_0(s_1)} \ket{\bar{0}0^{\tau}0} \ket{0^r}\|  \tag*{} \\
		&- \|\ket{g}\bra{g}_{1} \ket{1}\bra{1}_{\gamma_1  \cdots \gamma_r} U_{qfs} (s_r) \cdots  U_{qfs}(s_1)( \sum_{k=1}^{n-1} \beta_k^1 \ket{v_k(s_1)})\ket{\omega_0}\| \tag*{where $\Pi_g = \ket{g}\bra{g}_{1}$, $\ket{\omega_i} = \ket{\bar{0}0^{\tau}0}\ket{1^i 0^{r-i}}$ } \\
		&\ge \|\ket{g}\bra{g}_{1} \ket{1}\bra{1}_{\gamma_1 \cdots\gamma_r}U_{qfs}(s_r)\cdots \beta_0^1 \ket{v_0(s_1)}  \ket{\omega_1}\| - B_1 \tag*{} \\
		&\ge \|\ket{g}\bra{g}_{1} \ket{1} \bra{1}_{\gamma_1  \cdots \gamma_r} U_{qfs}(s_r) \beta_0^{r-1} \cdots \beta_0^1 \ket{v_0(s_r)} \ket{\omega_{r-1}} \| - B_1 - \cdots -B_r \tag*{} \\
		&= |\left \langle  g |  v_0(s_r) \right\rangle\beta_0^{r}\beta_0^{r-1}\cdots\beta_0^1| - B_1 -\cdots -B_r,\tag*{}
	\end{align}
	where	$B_i=	 \|\ket{g}\bra{g}_{1} \ket{1} \bra{1}_{\gamma_1  \cdots \gamma_r} U_{qff} (s_r) \cdots  U_{qff} (s_{i}) ( \sum_{k=0}^{n-1} \beta_k^i\beta_0^{i-1} \cdots \beta_0^1 \ket{v_k(s_i)})  \ket{\omega_{i-1}} \|$. Since the $i$th qubit in the fifth register changes from $\ket{0}$ to $\ket{1}$ only if the state before $ccX_{234,\gamma_i}$ is expressed as $\ket{\cdots}\ket{\bar{0}0^{\tau}0}\ket{1^{i-1}0^{r-i+1}}$, we have
	\begin{align} B_i
		=&	 \|\ket{g}\bra{g}_{1} \ket{1}\bra{1}_{\gamma_1 \cdots\gamma_r}U_{qfs}(s_r)\cdots U_{qfs}(s_{i})(\sum_{k=0}^{n-1} \beta_k^i\beta_0^{i-1} \cdots \beta_0^1 \ket{v_k(s_i)})\ket{\omega_{i-1}}\|\tag*{}\\
		\le&	 \|(\ket{1}\bra{1}_{\gamma_i})U_{qfs}(s_r)\cdots U_{qfs}(s_{i})(\sum_{k=0}^{n-1} \beta_k^i\beta_0^{i-1} \cdots \beta_0^1 \ket{v_k(s_i)}) \ket{\omega_{i-1}}\|\tag*{}\\
		=&	 \|U_{qfs}(s_r)\cdots (\ket{1}\bra{1}_{\gamma_i})U_{qfs}(s_{i})(\sum_{k=0}^{n-1} \beta_k^i\beta_0^{i-1} \cdots \beta_0^1 \ket{v_k(s_i)}) \ket{\omega_{i-1}}\|\tag*{}\\
		=&	 \|\ket{1}\bra{1}_{\gamma_i}U_{qfs}(s_{i})(\sum_{k=0}^{n-1} \beta_k^i\beta_0^{i-1} \cdots \beta_0^1 \ket{v_k(s_i)}) \ket{\omega_{i-1}}\|\tag*{}\\
		=&	 \|\ket{1} \bra{1}_{\gamma_i } ccX_{2,4} X_4 (U_{qff}^{ \dagger}  (W(s_i), \tau)\x I)ccX_{234,\gamma_i} (U_{qff} (W(s_i), \tau)\x I)( \sum_{k=0}^{n-1} \beta_k^i \beta_0^{i-1} \cdots \beta_0^1 \ket{v_k(s_i)}) \ket{\omega_{i-1}}\|\tag*{ }\\
		=&	 \|ccX_{2,4} X_4 (U_{qff}^{ \dagger}  (W(s_i), \tau)\x I) \ket{1} \bra{1}_{\gamma_i}  ccX_{234,\gamma_i} (U_{qff} (W(s_i), \tau)\x I) ( \sum_{k=0}^{n-1} \beta_k^i \beta_0^{i-1} \cdots \beta_0^1 \ket{v_k(s_i)}) \ket{\omega_{i-1}}\|\tag*{}\\
		=&	 \|\ket{1}\bra{1}_{\gamma_i} ccX_{234,\gamma_i}  (U_{qff}(W(s_i), \tau)\x I)(\sum_{k=0}^{n-1} \beta_k^i\beta_0^{i-1} \cdots \beta_0^1 \ket{v_k(s_i)}) \ket{\omega_{i-1}}\|\tag*{}\\
		=&	 \|\ket{\bar{0}}\bra{\bar{0}}_{2}\ket{0^\tau}\bra{0^\tau}_{3}\ket{0^\tau}\bra{0}_{4}(U_{qff}(W(s_i), \tau)\x I)(\sum_{k=0}^{n-1} \beta_k^i\beta_0^{i-1} \cdots \beta_0^1 \ket{v_k(s_i)}) \ket{\omega_{i-1}}\|\tag*{}\\
		\le&\| \Pi_{\bar{0}0} (U_{qff}(W(s_i), \tau)\x I) (\sum_{k=1}^{n-1} \beta_k^i \beta_0^{i-1} \cdots \beta_0^1  \ket{v_k(s_i)}) \ket{\omega_{i-1}}\|.\tag*{}
	\end{align}
	From (\ref{eq:9}) we have the above equation satisfies 
	$$B_i  =  |\beta_0^{i-1} \cdots \beta_0^1 |\cdot| \cos \theta_{i-1} - \sin \theta_{i-1} \frac{\cos\theta_i}{\sin\theta_i} |\cdot\| \Pi_{\bar{0}0}(U_{qff} (W(s_i), \tau)\x I)( \sum_{k=1}^{n-1}  \alpha_{k,0}^i \ket{v_k(s_i)} ) \ket{\omega_{i-1}} \|,$$
	then from triangle equality, we have
	\begin{align}
		&B_i \le |\beta_0^{i-1} \cdots \beta_0^1|\cdot| \cos \theta_{i-1} - \sin\theta_{i-1}\frac{\cos\theta_{i}}{\sin\theta_{i}}|\cdot(\| \sum_{k=1}^{n-1} (D^{t}\alpha_{k,0}^i\ket{v_k(s)})\ket{\omega_{i-1}}\|+\tag*{} \\
		&  \| \Pi_{\bar{0}0} (U_{qff}(W(s_i), \tau)\x I)\ket{\omega_{i-1}}\sum_{k=1}^{n-1} \alpha_{k,0}^i\ket{v_k(s_i)}\ket{\omega_{i-1}}-\sum_{k=1}^{n-1}(D^{t}\alpha_{k,0}^i\ket{v_k(s_i)})\ket{\omega_{i-1}}\|). \tag*{} 
	\end{align}
	From \Cref{lem:spectrum} we have
	$$\| D^{t}\ket{v_k(s)}\|=\cos^{t}\varphi_k \le \cos^{t}\varphi_1= (1 - \Delta(s))^t =  e^{-t \log(\frac{1}{1-\Delta(s)})}.$$
	Let $t = \log(\frac{2r}{\varepsilon})\frac{2}{\Delta(s)}$, where $\Delta(s)$ denotes the eigenvalue gap of $P(s)$, then
	$$\| D^{t}\ket{v_k(s)}\|= \left[e^{- \log(\frac{2r}{\varepsilon})} \right]^{\frac{2}{\Delta(s)}\log(\frac{1}{1-\Delta(s)})} = (\frac{\varepsilon}{2r})^{\frac{2}{\Delta(s)}\log(\frac{1}{1-\Delta(s)})} \le \frac{\varepsilon}{2r}.$$
	Combine with (\ref{eq:close}) and $\varepsilon_{1} = \frac{\varepsilon}{2r}$, the above equation can be expressed as 
	\begin{align*}
		&|\beta_0^{i-1} \cdots \beta_0^1|\cdot| \cos\theta_{i-1}  - \sin\theta_{i-1}\frac{\cos\theta_{i}}{\sin\theta_{i}}|\cdot(2 |\sum_{k=1}^{n-1} \cos^{t}(\varphi_k)\alpha_{k,0}^i|  + \frac{\varepsilon}{2r})\tag*{}\\
		\le&|\beta_0^{i-1} \cdots \beta_0^1|\cdot|\cos\theta_{i-1}  - \sin\theta_{i-1}\frac{\cos\theta_{i}}{\sin\theta_{i}}|\cdot(\frac{\varepsilon}{2r}| \sum_{k=1}^{n-1}\alpha_{k,0}^i|  + \frac{\varepsilon}{2r})	\tag*{}\\
		\le&|\beta_0^{i-1} \cdots \beta_0^1|\cdot|\cos\theta_{i-1} - \sin\theta_{i-1}\frac{\cos\theta_{i}}{\sin\theta_{i}}|\cdot\frac{\varepsilon}{r}\le\varepsilon.
	\end{align*}
	From (\ref{eq:theta}) we have
	\begin{align*}
		\sqrt{p_{succ}}&= |\left \langle  g |  v_0(s_r) \right\rangle\beta_0^{r}\beta_0^{r-1}\cdots\beta_0^1| - B_1 - B_2-\cdots -B_r\\
		&\ge \cos^{k+1}(\frac{\pi}{2(k+1)})-\varepsilon. 
	\end{align*}	
	From the definition of $f(r)$ in proof of \Cref{thm:1}, success probability increases with the number of steps $r$ increases, so the algorithm avoids the souffl\'e problems that amplitude amplification will bring.
	Compared with the result in \Cref{thm:1}, we achieve a faster search algorithm with less ancilla qubits.
	\end{proof}
	In general, for a reversible Markov chain, if we reverse a search algorithm, we can obtain a qsamping algorithm.
	Reversing the above algorithm results, we get a qsampling algorithm as shown below.
	\begin{theorem}
	[Qsampling based on generalized interpolated walks] \label{thm:3}
		For an ergodic reversible Markov chain P in $G(V,E)$, there exists an algorithm that outputs $\ket{\pi}$ with success probability more than $\frac{4}{5}$ from initial state $\ket{g}$ for any $g \in V$. The complexity is $\Theta(\Seg(g)+\log(\varepsilon^{-1})\sqrt{\Heg} (\Ueg +  \Ceg))$, with $\Theta(\log\log(\varepsilon^{-1})+\log\sqrt{\Heg})$ ancilla qubits, where $\varepsilon$ is the error.
	\end{theorem}
	\subsection{Application}\label{sect:Application}
	In this section, we give an application of generalized interpolated walks. Adiabatic quantum computing is an important quantum computing model that is equivalent to standard quantum circuit model \cite{Aharonov2003AdiabaticQS}. We apply generalized interpolated walks to prepare quantum stationary state of Markov chains for adiabatic quantum computing. 
	In adiabatic quantum computing, the state preparation is the premise of algorithm and is always expressed as a ground state of the target Hamiltonian which can be obtained by a series of slowly changing Hamiltonian.
	
	However, the slowly evolving Hamiltonian sequence is not always easy to prepare. For those Hamiltonian that corresponds the reversible Markov chains, the corresponding ground state  can be equivalent to  the stationary state \cite{Aharonov2003AdiabaticQS}, and the construction of slowly evolving Hamiltonian sequences are reduced to construct a series of slowly evolving Markov chains with their stationary states \cite{Wocjan2008SpeedupVQ}, where the algorithm to prepare Boltzmann-Gibbs distribution has been constructed. However, the construction of general slowly evolving Markov chain sequence is not given in \cite{Wocjan2008SpeedupVQ}, although it is necessary when we want to prepare stationary state of any reversible Markov chain.
	
	By applying generalized interpolated walks, we give the detailed construction of slowly evolving Markov chains to approach any reversible Markov chain P as follows.
	
	Let $Q=(q_{i,j})_{1\le i,j\le r}$  satisfy $$q_{i,j}= \left\{\begin{matrix}
		1 \quad if \quad j = i+1,1\le i\le r-1\\
		0 \quad \quad \quad \quad \quad \quad \quad \ \quad otherwise
	\end{matrix}\right.$$
	We define a series of Markov chains as:
	$$P(s_0),P(s_1), P(s_2), \cdots, P(s_i), \cdots, P(s_r).$$
	The walk operator can be defined as
	$$P(s_i) \coloneqq (1-s_i)P + s_i P',\ W(s_i) \coloneqq V(s_i)^{\dagger} \cdot S \cdot V(s_i) \cdot R_{\bar{0}},$$
	and the corresponding details can be seen in \Cref{sect:preliminaries}.
	Let $\pi(s_i) =  \frac{(1-s_i)\pi+s_i\pi'}{1-s_i+s_i\pi_0}$.
	From proposition 11 in \cite{Krovi2015QuantumWC}, for every $s\in[0,1)$, $P(s_i)$ is a reversible Markov chain and $\pi(s_i)$ is the only stationary distribution of $P(s_i)$. Define the corresponding quantum state as
	$\ket{\pi(s_i)}=\sum_{x}\sqrt{\frac{(1-s_i)\pi_x+s_i\pi'_x}{1-s_i+s_i\pi_0}}\ket{x}$,
	we have
	\begin{align*}	W(s_i)\ket{\pi(s_i)}\ket{\bar{0}} 
		=&V(s_i)^{\dagger} \cdot S \cdot V(s_i) \cdot R_{\bar{0}}(\sum_{x,y} \sqrt{\pi(s_i)_x}\ket{x}\ket{\bar{0}})\\
		=&V(s_i)^{\dagger} \cdot S(\sum_{y}\sqrt{\pi(s_i)_y}(\sum_{x}\sqrt{P(s_i)_{y,x}}\ket{x})\ket{y})\\
		=&V(s_i)^{\dagger}(\sum_{y}\sqrt{\pi(s_i)_y}\ket{y}(\sum_{x}\sqrt{P(s_i)_{y,x}}\ket{x}))\\
		=&V(s_i)^{\dagger}V(\sum_{y}\sqrt{\pi(s_i)_y}\ket{y}\ket{\bar{0}})
		=\ket{\pi(s_i)}\ket{\bar{0}},
	\end{align*}
	that is for any $s\in[0,1)$, $\ket{\pi(s_i)}\ket{\bar{0}}$ is 1-eigenvector of $W(s_i)$.
	
	According to the construction of $\ket{\pi(s_i)}$ we know:
	\begin{equation}
		\ket{\pi(s_i)} = \sqrt{\frac{(1-s_i)(1-\pi_0)}{1-s(1-\pi_0)}} \ket{\bar{\pi}} + \sqrt{\frac{\pi_0}{1-s(1-\pi_0)}} \ket{0}.
	\end{equation}
	For any $i \in \{0, 1, \cdots, r\}$, let
	$$s_i=\frac{1-\pi_0/\sin^2\theta_i}{1-\pi_0}=\frac{1-\pi_0/\sin^2(\frac{i\pi}{2(r+1)})}{1-\pi_0}, i \in \{0, 1, \cdots, r\},$$
	and we have
	$\ket{v_0(s_i)}=
	\cos\theta_i\ket{\bar{\pi}}+\sin\theta_i\ket{0}=\cos(\frac{i\pi}{2(r+1)})\ket{\bar{\pi}}+\sin(\frac{i\pi}{2(r+1)})\ket{0}$.

	Now we prove that all above Markov chains satisfies varying slowing condition, which means that for any two adjacent Markov chains, there exists $q$ such that the inner product of their quantum stationary state is greater than $q$.
	
	Select the initial quantum stationary state to be $\ket{0}\in\spn\{\ket{0},\ket{1},\cdots,\ket{n-1}\}$, which is easy to prepare. Since we have $\left \langle \bar{\pi}| 0 \right \rangle =0$, which means \begin{align*}
		\left \langle \pi_i | \pi_{i+1} \right \rangle 
		=&(\cos\theta_i\bra{\bar{\pi}}+\sin\theta_i\bra{0})(\cos\theta_{i+1}\ket{\bar{\pi}}+\sin\theta_{i+1}\ket{0})\\
		=&\cos\theta_i\cos\theta_{i+1}+\sin\theta_i\sin\theta_{i+1}
		=\cos(\frac{\pi}{2(r+1)}).
	\end{align*}
	If $r\ge\frac{\pi}{2\arccos q}-1$, then  $\cos(\frac{\pi}{2(r+1)})\ge q$, and the slowing evolving condition $\left \langle \pi_i | \pi_{i+1} \right \rangle \ge q$ is satisfied for $i \in \{0, 1, \cdots, r\}$, where the quantum stationary state of the first Markov chain is easy to prepare and the quantum stationary state of the last Markov chain approximates target state.
	
	\section{Discussion} \label{sect:Discussion}
	\hspace{1.2em}
	In this work, by defining and applying generalized interpolated walks instead of amplitude amplification, we both improve the success probability of search algorithm to nearly 1, and avoid the souffl\'e problems that amplitude amplification will bring. By introducing quantum fast-forwarding our new algorithms not only reduce the dependency on error but also reduce the number of ancilla qubits required. In this process, we find the relationship between the phase estimation  and quantum fast-forwarding, as well as the relationship between the generalized interpolated walks and amplitude amplification. We think there is something that deserves to be explored in the future.
	
	Besides, we use generalized interpolated walks to improve the success probability of qsampling algorithm as well.
	Finally, by applying generalized interpolated walks to adiabatic quantum computing, we construct a series of slowing evolving Markov chains, which satisfy that the first stationary state is easy to prepare and the final stationary state is the target state. The process is fundamental in adiabatic stationary state preparation and we give the detailed construction.
	
	\section*{Acknowledgements}
	We thank the support of National Natural Science Foundation of China (Grant No.61872352), and Program for Creative Research Group of National Natural Science Foundation of China (Grant No. 61621003).

\end{document}